\documentclass[a4paper,UKenglish,cleveref]{lipics-v2021}

\pdfoutput=1 
\hideLIPIcs  

\usepackage{mathtools}
\usepackage{tikz-cd}
\usepackage{caption}
\usetikzlibrary{decorations.markings,decorations.pathreplacing,
  shapes.geometric,matrix,arrows,chains,positioning,scopes}

\pgfdeclarearrow{
  name = pxto,
  setup code = {
    \pgfarrowssettipend{1.5\pgflinewidth}
    \pgfarrowssetbackend{-2.5508\pgflinewidth}
    \pgfarrowssetlineend{-.25\pgflinewidth}
    \pgfarrowssetvisualbackend{-0.021\pgflinewidth}
    \pgfarrowsupperhullpoint{1.5\pgflinewidth}{0\pgflinewidth}
    \pgfarrowsupperhullpoint{-2.0085\pgflinewidth}{3.6525\pgflinewidth}
    \pgfarrowsupperhullpoint{-2.5508\pgflinewidth}{3.0763\pgflinewidth}
  },
  drawing code = {
    \pgfsetdash{}{0pt}%
    \pgfpathmoveto{\pgfpoint{1.5\pgflinewidth}{0.0254\pgflinewidth}}%
    \pgfpathlineto{\pgfpoint{-2.0085\pgflinewidth}{3.6525\pgflinewidth}}%
    \pgfpathlineto{\pgfpoint{-2.5508\pgflinewidth}{3.0763\pgflinewidth}}%
    \pgfpathlineto{\pgfpoint{-0.4322\pgflinewidth}{0.5\pgflinewidth}}%
    \pgfpathlineto{\pgfpoint{-0.4322\pgflinewidth}{-0.5\pgflinewidth}}%
    \pgfpathlineto{\pgfpoint{-2.5508\pgflinewidth}{-3.0763\pgflinewidth}}%
    \pgfpathlineto{\pgfpoint{-2.0085\pgflinewidth}{-3.6525\pgflinewidth}}%
    \pgfpathclose%
    \pgfusepathqfill
  }
}
\tikzset{>=pxto}
\tikzcdset{arrow style=tikz}

\newcommand{\agdaunimath}{\textsc{agda-unimath}}
\newcommand{\TypeTopology}{\textsc{TypeTopology}}

\DeclareMathOperator{\id}{id}
\DeclarePairedDelimiter{\pa}{(}{)}
\newcommand{\SigmaT}[2]{\sum\pa*{#1}\,{#2}}
\newcommand{\PiT}[2]{\prod\pa*{#1}\,{#2}}
\newcommand{\lambdadot}[2]{\lambda{#1}.\,{#2}}
\DeclareMathOperator{\susp}{\Sigma}
\DeclareMathOperator{\eval}{eval}
\DeclareMathOperator{\North}{N}
\DeclareMathOperator{\South}{S}
\DeclareMathOperator{\merid}{merid}

\newcommand{\colonequiv}{\mathrel{\vcentcolon\mspace{-1mu}\equiv}}
\DeclareMathOperator{\consts}{consts}
\DeclareMathOperator{\fst}{pr_1}
\DeclareMathOperator{\snd}{pr_2}
\newcommand{\One}{\mathbf 1}
\DeclareMathOperator{\ap}{ap}
\DeclareMathOperator{\inl}{inl}
\DeclareMathOperator{\inr}{inr}

\newcommand{\pathcomp}{%
  \mathchoice{\mathbin{\raisebox{0.5ex}{$\displaystyle\centerdot$}}}%
             {\mathbin{\raisebox{0.5ex}{$\centerdot$}}}%
             {\mathbin{\raisebox{0.25ex}{$\scriptstyle\,\centerdot\,$}}}%
             {\mathbin{\raisebox{0.1ex}{$\scriptscriptstyle\,\centerdot\,$}}}
}

\newcommand{\mkTTurl}[1]{\href{https://www.cs.bham.ac.uk/~mhe/TypeTopology/#1.html}{\texttt{#1}}}

\title{Formalizing equivalences without tears}

\author{Tom {de Jong}}%
{School of Computer Science, University of Nottingham, United Kingdom %
  \and \url{http://www.tdejong.com}}
{tom.dejong@nottingham.ac.uk}
{https://orcid.org/0000-0003-1585-3172}
{}

\authorrunning{T. de Jong}

\Copyright{Tom de Jong}



\category{} 

\relatedversion{} 

\funding{%
  This work was supported by The Royal Society (grant reference
  URF\textbackslash R1\textbackslash191055).
  The HIM programme was funded by the Deutsche Forschungsgemeinschaft (DFG, German
  Research Foundation) under Germany's Excellence Strategy -- EXC-2047/1 --
  390685813.}

\acknowledgements{%
  I am grateful to Josh Chen, Stefania Damato, Mart\'in Escard\'o, Nicolai Kraus,
  Fredrik Nordvall Forsberg, Sti\'ephen Pradal and Jon Sterling for comments on
  this note.
  I also thank the participants, organizers and support staff of the HIM trimester
  programme \emph{Prospects of Formal Mathematics} for the opportunity to present
  the ideas in this note.
  Finally, I thank the anonymous reviewers for their helpful suggestions.}

\nolinenumbers 

\EventEditors{Rasmus Ejlers M\o{}gelberg and Benno van den Berg}
\EventNoEds{2}
\EventLongTitle{30th International Conference on Types for Proofs and Programs (TYPES 2024)}
\EventShortTitle{TYPES 2024}
\EventAcronym{TYPES}
\EventYear{2024}
\EventDate{June 10--14, 2024}
\EventLocation{Copenhagen, Denmark}
\EventLogo{}
\SeriesVolume{336}
\ArticleNo{1}

\begin{document}

\maketitle

\begin{abstract}
  This expository note describes two convenient techniques in the context of
  homotopy type theory for proving---and formalizing---that a given map is an
  equivalence.
  The first technique decomposes the map as a series of basic equivalences,
  while the second refines this approach using the 3-for-2 property of
  equivalences.
  The techniques are illustrated by proving a basic result in synthetic
  homotopy theory.
\end{abstract}

\section{Introduction}

A very common problem in \emph{homotopy type theory (HoTT)}~\cite{HoTTBook} is
to prove that a given map is an equivalence.
The purpose of this short note is to describe convenient techniques for doing
this, in particular when one is interested in formalizing the argument in a
proof assistant.
I claim no originality in the results of this note. Indeed, the technique I wish
to highlight already informs much of the \agdaunimath\ library developed by
Rijke and contributors~\cite{agda-unimath}, while I picked up the other
(decomposition) technique in this note via the Agda development \TypeTopology\
of Escard\'o and collaborators~\cite{TypeTopology} as well as Escard\'o's
comprehensive introduction to univalent foundations and its formalization in
Agda~\cite{Escardo2019}.
Rather, I hope that this note will contribute to a greater awareness
of these techniques, especially among junior type theorists.

\subsection*{Outline}
The note is outlined as follows. \cref{sec:naive} explains why directly proving
that a map is an equivalence is often cumbersome.
\cref{sec:decomposition} describes an alternative technique by decomposing the
given map into a sequence of (smaller) equivalences, while \cref{sec:3-for-2}
further refines this method using the fact that equivalences satisfy the
\emph{3-for-2 property}.
These techniques are then illustrated in \cref{sec:example} in the context of
\emph{synthetic homotopy theory}~\cite[\S8]{HoTTBook}. The example is
self-contained and no prior knowledge of this area is required, although someone
familiar with classical homotopy theory may find its simplicity appealing and
consider it an invitation to learn more (see e.g.~\cite{Shulman2021}).

\subsection*{Terminology}
The 3-for-2 property of equivalences states that if any two maps in a
commutative triangle are equivalences, then so is the third.
This property is perhaps more commonly known as the \emph{2-out-of-3
  property}. Andr\'e Joyal proposed the name 3-for-2 by analogy to how discounts
are often advertised; if you prove two maps are equivalences, then the third is
for free~\cite{Joyal2013}.
Since we are interested in reducing the amount of (formalization) work and wish
to get as much as possible for free, the analogy is quite apt and we prefer
Joyal's terminology in this note.

\subsection*{Foundations}
As mentioned at the very start, this note is concerned with equivalences in
homotopy type theory, although the issues and techniques described should carry
over to other intensional type theories~\cite{Hofmann1955} and classes of maps
that satisfy 3-for-2.\footnote{However, I don't have a good illustration to hand
  of applying the techniques to a class of maps other than the equivalences. The
  class of \(n\)-equivalences, i.e.\ those maps whose \(n\)-truncation is an
  equivalence, comes to mind, but there I have found it easier to work with the
  fact that these maps can be characterized as those maps for which
  precomposition into \(n\)-types is an equivalence~\cite[Lem.~2.9]{CORS2020}
  (see
  also~\cite[\href{https://unimath.github.io/agda-unimath/foundation.truncation-equivalences.html\#properties}{\texttt{k-equivalences}}]{agda-unimath}).
  Moreover, while equivalences are invertible, maps satisfying 3-for-2 need not
  be of course, which means that it may be more difficult to arrange a diagram
  like in~\eqref{decomposition'}.}
We mostly adopt the terminology and notation of the HoTT Book~\cite{HoTTBook},
e.g.\ writing \(\equiv\) for judgemental (definitional) equality, using \(=\) for
identity types (sometimes known as propositional equality), and \(\sim\) for
homotopies (i.e.\ pointwise identities).

\section{A naive approach}\label{sec:naive}

Presented with the problem of showing that a map \(f : A \to B\) is an
equivalence,\footnote{The precise definition of an equivalence is somewhat
  subtle, as discussed at length in \cite[\S4]{HoTTBook}, but it is not too
  important for our purposes.} a direct approach would be to try and construct
a map \(g : B \to A\) together with identifications \({g \circ f} \sim \id_A\)
and \({f \circ g} \sim \id_B\).
What makes this approach infeasible at times is that the construction of \(g\)
may be involved, resulting in nontrivial computations (often involving
\emph{transport}) when showing that the round trips are homotopic to the identity
maps, especially in proof-relevant settings such as HoTT.

In some cases we are lucky and the desired identifications hold definitionally,
in which case the proof assistant can simply do the work for us by unfolding
definitions. We return to using definitional equalities to our advantage in
\cref{sec:3-for-2}.

\section{Decomposition into equivalences}\label{sec:decomposition}

Instead of directly arguing that a given map \(f : {A \to B}\) is an
equivalence, it is often convenient to instead decompose \(f\) as a series of
``building block equivalences'', general maps that we already know to be
equivalences, as depicted below.
\begin{equation}\label{decomposition}
  \begin{tikzcd}
    A \ar[dr,"\simeq",bend right] \ar[rrrrr,"f"] & & & & & B \\
    & X_1 \ar[r,"\simeq"] & X_2 \ar[r,"\simeq"]
    & \dots \ar[r,"\simeq"] & X_n \ar[ur,"\simeq",bend right]
  \end{tikzcd}
\end{equation}

Some quintessential examples of such building block equivalences are as follows.
\begin{description}
\item[Projection from the total space of a contractible family.] Given a type
  \(X\) and a dependent type \(Y\) over it, if each \(Y(x)\) is contractible
  (i.e.\ it is equivalent to the unit type), then the projection map
  \(\fst : \SigmaT{x : X}Y(x) \to X\) is an equivalence.
  In fact, this is an equivalence if and only if each \(Y(x)\) is contractible.
\item[Contractibility of singletons.]\label{proj-equivalences} For any type \(X\)
  and \(x : X\), the type \(\SigmaT{y : X}{x = y}\) is contractible and hence
  the two projection maps from \(\SigmaT{x : X}{\SigmaT{y : X}{x = y}}\) to
  \(X\) are equivalences.
\item[Associativity of dependent sums.] For a type \(A\), and dependent types
  \(B(a)\) and \(C(a,b)\) over \(A\) and \(\SigmaT{a : A}{B(a)}\), respectively,
  the map
  \begin{align*}
    \SigmaT{a : A}{\SigmaT{b : B(a)}{C(a,b)}}
    &\quad\to\quad \SigmaT{p : \SigmaT{a : A}{B(a)}}{C(p)} \\
    (a,b,c) &\quad\mapsto\quad ((a,b) , c)
  \end{align*}
  is an equivalence.
\item[Reindexing dependent sums along an equivalence.] Given an equivalence
  between types \({f : {A \simeq B}}\) and a dependent type \(Y(b)\) over \(B\),
  the assignment
  \begin{align*}
    \SigmaT{a : A}Y(f(a))
    &\quad\to\quad \SigmaT{b : B}{Y(b)} \\
    (a,y) &\quad\mapsto\quad (f(a),y)
  \end{align*}
  is an equivalence.
  And of course there is a similar result for dependent products.
\item[Distributivity of \(\prod\) over \(\sum\).] Given a type \(A\) and
  dependent types \(B\) and \(Y\) over \(A\) and \(\SigmaT{a : A}{B(a)}\),
  respectively, the map
  \begin{align*}
    \PiT{a : A}{\SigmaT{b : B(a)}Y(a,b)} &\quad\to\quad \SigmaT{f : \PiT{a : A}{B(a)}}{\PiT{a : A}Y(a,f(a))} \\
    \alpha &\quad\mapsto\quad (\lambdadot{a}{\fst(\alpha(a)}),\lambdadot{a}{\snd(\alpha(a))})
  \end{align*}
  is an equivalence with inverse \((f,p) \mapsto \lambdadot{a}(f(a),p(a))\).

  This, and especially its nondependent version
  \[
    \PiT{a : A}{\SigmaT{b : B}{Y(a,b)}} \to \SigmaT{f : A \to B}{\PiT{a :
        A}{Y(a,f(a))}},
  \] is traditionally called the ``type theoretic axiom of choice'', but this is
  a misnomer as there is no choice involved, see also~\cite[pp.~32 and 104]{HoTTBook}.
\item[Distributivity of \(\sum\) over \(+\).] Given a type \(A\) and dependent
  types \(X\) and \(Y\) over it, the map
  \begin{align*}
    \SigmaT{a : A}{(X(a) + Y(a))} &\quad\to\quad \SigmaT{a : A}{X(a)} \,+\, \SigmaT{a : A}{Y(a)} \\
    (a,\inl x) &\quad\mapsto\quad \inl(a,x) \\
    (a,\inr y) &\quad\mapsto\quad \inr(a,y)
  \end{align*}
  is an equivalence.
\item[Congruence of type formers.] All type formers respect equivalences. For
  example, if we have \(f : A \simeq X\) and \(g : B \simeq Y\), then
  \((A + B) \simeq (X + Y)\) by applying \(f\) to the elements on the left and \(g\)
  to the elements on the right.
\item[Composition with a fixed path.] Given elements \(x\), \(y\) and \(z\) of a
  type \(X\) and a path \(p_0 : x = y\), the path composition maps

  \begin{minipage}{0.4\linewidth}
    \begin{align*}
      (z = x) &\quad\to\quad (z = y) \\
      p &\quad\mapsto\quad p \pathcomp p_0
    \end{align*}
  \end{minipage}
  and\hspace{1cm}
  \begin{minipage}{0.5\linewidth}
    \begin{align*}
      (x = z) &\quad\to\quad (y = z) \\
      p &\quad\mapsto\quad p_0 \pathcomp p
    \end{align*}
  \end{minipage}\vspace{1em}

  are equivalences.
\end{description}

This list is not exhaustive, but should give a good impression of the available
building blocks. The interested reader may find many more examples in
\cite[{\mkTTurl{UF.EquivalenceExamples}}]{TypeTopology}.

The idea is to reduce the task of proving that \(f\) is an equivalence to
identifying suitable building blocks---the equivalences in
\eqref{decomposition}---\emph{and proving that the diagram \eqref{decomposition}
  commutes}.
It is the latter point that may pose similar difficulties to those explained in
the previous section: the commutativity proof could involve nontrivial
computations. We turn to a refinement in the next section to address this.

We should mention that this technique is still very valuable, especially when we
are not interested in having a particular equivalence, or when there is a unique
such equivalence, e.g.\ when proving that a type \(X\) is contractible.

\section{A refinement using 3-for-2}\label{sec:3-for-2}
To address the issue identified above, we will make use of the fact that
equivalences satisfy the 3-for-2 property:

\begin{lemma}[3-for-2 for equivalences, {\cite[Thm.~4.7.1]{HoTTBook}}]
  In a commutative triangle
  \[
    \begin{tikzcd}
      A \ar[dr,"f"'] \ar[rr,"h"] & & C \\
      &B \ar[ur,"g"']
    \end{tikzcd}
  \]
  if two of the maps are equivalences, then so is the third.
\end{lemma}

Rather than decomposing \(f\), the idea is to simply involve \(f\) into
\emph{any} commutative diagram where all (other) maps are equivalences, as
depicted below.

\begin{equation}\label{decomposition'}
  \begin{tikzcd}[column sep=.5cm]
    X_1 \ar[rrrrrrrr,bend right=3ex,"\simeq",pos=.55,
    start anchor=south east,end anchor=south west]
    \ar[r,"\simeq"]
    & X_2 \ar[r,"\simeq"]
    & \dots \ar[r,"\simeq"] & X_i \ar[r,"\simeq"]
    & A \ar[r,"f"] & B \ar[r,"\simeq"]
    & X_{i+1} \ar[r,"\simeq"]
    & \dots \ar[r,"\simeq"]
    & X_n
  \end{tikzcd}
\end{equation}

By the 3-for-2 property, we can still conclude that \(f\) is an equivalence from
\eqref{decomposition'}, but verifying the commutativity may now be easier,
especially when the type \(X_n\) is simpler than~\(B\).
Indeed, we can often ensure this in practice as illustrated and commented on at
the end of the next section.

\section{An example in synthetic homotopy theory}\label{sec:example}

Usually~\cite[\S7.5]{HoTTBook}, a type is said to be \emph{\(n\)-connected}
if its \(n\)-truncation is contractible. For us, the following
characterization may serve as a definition.%

\begin{proposition}[{\cite[Cor.~7.5.9]{HoTTBook}}]%
  \label{connectedness-characterization}
  A type \(A\) is \(n\)-connected if and only if for every \(n\)-type%
  \footnote{We recall from~\cite[\S7.1]{HoTTBook} that the \(n\)-types are inductively defined for
\(n\ge-2\): a \(-2\)-type is a contractible type and an \((n+1)\)-type is a type
whose identity types are \(n\)-types.}
  \(B\), the
  constants map
  \begin{align*}
    B &\to (A \to B) \\
    b &\mapsto \lambdadot{a}{b}
  \end{align*}
  is an equivalence.
\end{proposition}

We will illustrate the above techniques by giving a slick proof of a well-known
result (see e.g.~\cite[Thm.~8.2.1]{HoTTBook}) in homotopy theory: taking the
suspension of a type increases its connectedness by one.
The reader may wish to compare the proof below to that given in \emph{op.~cit},
or inspects its formalization as part of the \agdaunimath\
library~\cite[{\href{https://unimath.github.io/agda-unimath/synthetic-homotopy-theory.suspensions-of-types.html\#suspensions-increase-connectedness}{\texttt{Suspensions increase connectedness}}}]{agda-unimath}.

For completeness, we recall suspensions in homotopy type theory.

\begin{definition}[Suspension \(\susp\)]
  The \emph{suspension} \(\susp A\) of a type \(A\) is the pushout of the span
  \(\One \leftarrow A \rightarrow \One\). Equivalently, it is the higher
  inductive type generated by two point constructors \(\North,\South : \susp A\)
  (short for \emph{North} and \emph{South}) and a path constructor
  \(\merid : {A \to \North = \South}\) (short for \emph{meridian}).
\end{definition}

\begin{figure}[h]
  \centering
  \includegraphics[scale=0.155]{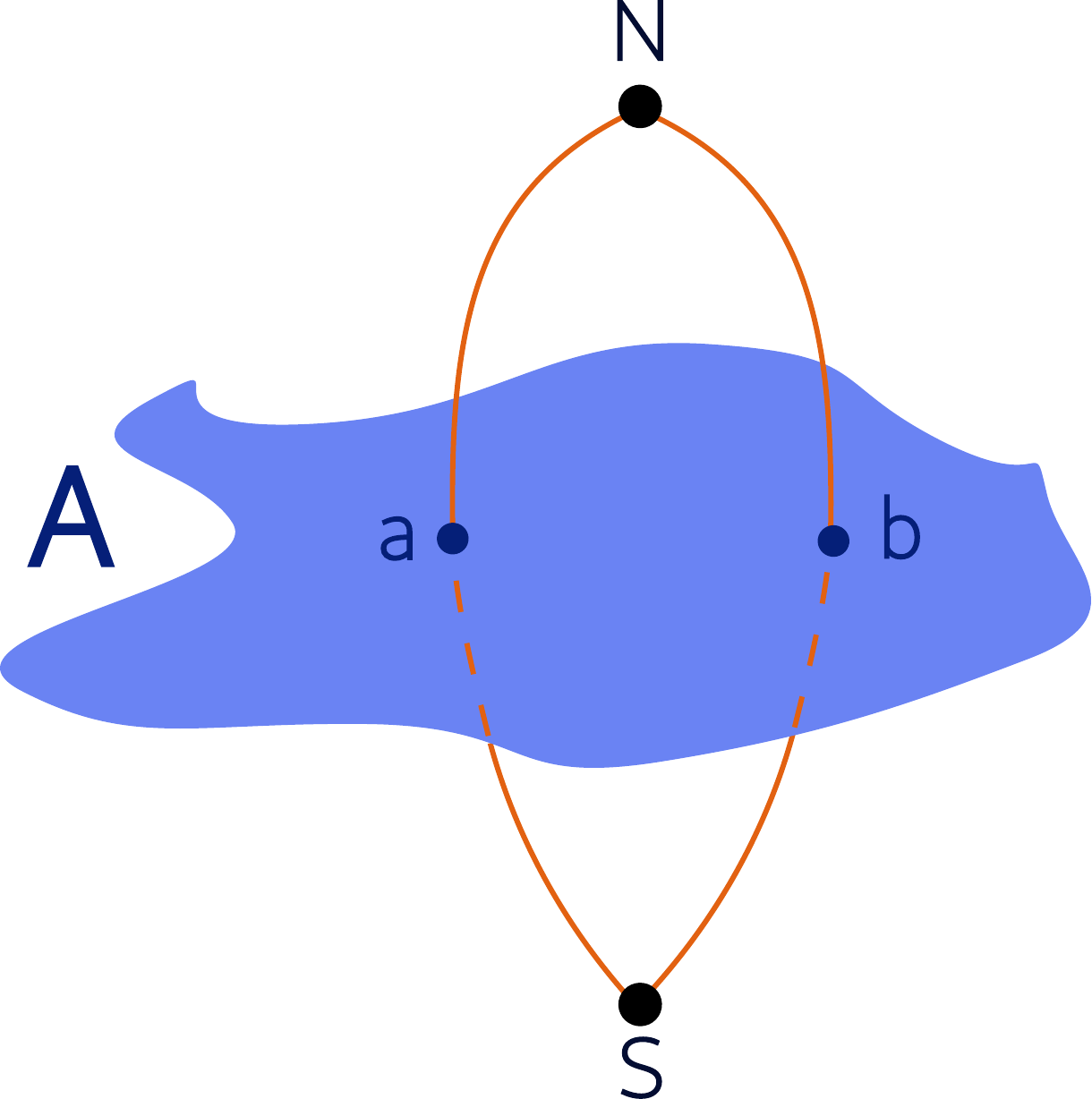}
  \caption*{Illustration of the suspension of a type \(A\).
    The lines from \(\North\) to \(\South\) going through the points \(a : A\) and
    \(b : A\) represent the paths \(\merid(a)\) and \(\merid(b)\),
    respectively.}
\end{figure}

We shall only need the following basic fact about suspensions which is just the
universal property of the suspension as a pushout.
(We recall from~\cite[\S2.2]{HoTTBook} that \(\ap_g\) denotes the action of \(g\)
on the identity types.)

\begin{proposition}[Universal property of the suspension, {\cite[Exer.~6.11]{HoTTBook}}]%
  \label{suspension-universal-property}
  \leavevmode

  \noindent The map
  \begin{align*}
    (\susp A \to B)
    &\quad\to\quad
    \SigmaT{b_N : B}{\SigmaT{b_S : B}{\pa*{A \to b_N = b_S}}} \\
    g
    &\quad\mapsto\quad \pa*{g(\North) , g(\South) , \lambdadot{a}{\ap_g(\merid(a))}}
  \end{align*}
  is an equivalence for all types \(A\) and \(B\).
\end{proposition}

\begin{theorem}
  The suspension of an \(n\)-connected type is \((n+1)\)-connected.
\end{theorem}
\begin{proof}
  Let \(A\) be an \(n\)-connected type and \(B\) an \((n+1)\)-type. By
  \cref{connectedness-characterization}, we need to show that the constants map
  \(\consts : {B \to (\susp A \to B)}\) is an equivalence.
  We first consider the evaluation map
  \[
    \eval : (\susp A \to B) \to B
    \quad\text{defined by}\quad
    \eval(g) \colonequiv g(\North),
  \]
  and note that the diagram
  \[
    \begin{tikzcd}[column sep=1cm]
      B \ar[rr,bend right,"\id_B"] \ar[r,"{\consts}"]
      & (\susp A \to B) \ar[r,"{\eval}"] & B
    \end{tikzcd}
  \]
  commutes definitionally.
  Hence, by \textbf{3-for-2}, it suffices to prove that \(\eval\)
  is an equivalence.
  We use the \textbf{decomposition} technique to do so and consider the
  following diagram.
  \begin{equation}\label{example-3-for-2-diagram}
    \begin{tikzcd}[column sep=0cm,row sep=1cm]
      (\susp A \to B)
      \ar[rrr,"\eval"]
      \ar[dr,"{g \;\mapsto\; (g(N),g(S),\dots)}","\simeq"']
      & & & B \\
      & \hspace{-1.8cm}\SigmaT{b_N : B}{\SigmaT{b_S : B}{(A \to b_N = b_S)}}
      \ar[r,"{(b_N,b_S,\varphi) \;\mapsto\; (b_N,b_S,\dots)}"',"\simeq",bend right]
      & \hspace{.2cm}\SigmaT{b_N : B}{\SigmaT{b_S : B}{(b_N = b_S)}}\hspace{-.3cm}
      \ar[ur,"{(b_N,b_S,\dots) \;\mapsto\; b_N}","\simeq"']
    \end{tikzcd}
  \end{equation}
  Note that because \(A\) is \(n\)-connected and \(b_N = b_S\) is an \(n\)-type
  (since \(B\) is an \((n+1)\)-type), the constants map
  \((b_N = b_S) \to (A \to b_N = b_S)\) is an equivalence.
  The middle equivalence in the diagram~\eqref{example-3-for-2-diagram} is
  induced by the inverse of this map.
  It should be stressed that the definitional behaviour of this inverse is
  completely irrelevant for verifying the commutativity of the
  diagram~\eqref{example-3-for-2-diagram}; we only need to know the middle map's
  behaviour on the first two components of the \(\Sigma\)-types to see that the
  diagram commutes definitionally.

  Finally, the first map in the decomposition of \(\eval\) is an equivalence by
  \cref{suspension-universal-property} and the last map is an equivalence by
  \textbf{contractibility of singletons} (from
  page~\pageref{proj-equivalences}).
\end{proof}

We end this note by pointing out an important \textbf{heuristic} that is nicely
illustrated by the above proof:
\emph{We always try to orient the maps in our diagrams towards the simplest
possible type}---which in \eqref{example-3-for-2-diagram} is clearly \(B\).
The reason for this is that checking commutativity should then be the easiest as
there is relatively little data in the target type to compare.
Similarly, we defined a section of \(\eval\), rather than directly defining a
map \(B \to (\susp A \to B)\) and proving it's an inverse to \(\eval\) which
would involve the cumbersome task of comparing functions \(\susp A \to B\) for
equality.

\bibliography{references}

\end{document}